\documentclass[sigconf]{acmart}
\AtBeginDocument{%
  }

\copyrightyear{2026}
\acmYear{2026}
\setcopyright{cc}
\setcctype{by}
\acmConference[KDD 2026] {Proceedings of the 32nd ACM SIGKDD Conference on Knowledge Discovery and Data Mining V.1}{August 9--13, 2025}{Jeju Island, Republic of Korea.}
\acmBooktitle{Proceedings of the 32nd ACM SIGKDD Conference on Knowledge Discovery and Data Mining V.1 (KDD 2026), August 9--13, 2025, Jeju Island, Republic of Korea}
\acmISBN{979-8-4007-2258-5/2026/08}
\acmDOI{10.1145/3770854.3783925}

\usepackage[ruled,vlined,noend]{algorithm2e}
\usepackage{enumitem}
\usepackage{multirow}
\usepackage{booktabs}
\usepackage{xspace}
\newcommand{\method}{MTGenRec\xspace}
\newcommand{\stitle}[1]{\vspace*{0.4em}\noindent{\bf #1\/}}
\usepackage{pifont}
\newcommand{\squishlist}{
	\begin{list}{$\bullet$}
		{ \setlength{\itemsep}{1pt}
			\setlength{\parsep}{1pt}
			\setlength{\topsep}{2.5pt}
			\setlength{\partopsep}{0.5pt}
			\setlength{\leftmargin}{1em}
			\setlength{\labelwidth}{1em}
			\setlength{\labelsep}{0.6em}
		}
	}
	\newcommand{\squishend}{
	\end{list}
}
\usepackage{graphicx}
\usepackage{fancyhdr}
\usepackage{balance}

\begin{document}

\title{MTGenRec: An Efficient Distributed Training System for \\ Generative Recommendation Models in Meituan}

\author{Yuxiang Wang}
\authornote{Both authors contributed equally to this research.}
\affiliation{%
  \institution{Wuhan University}
  \department{School of Computer Science}
  \city{Wuhan}
  \country{China}
}
\email{nai.yxwang@whu.edu.cn}

\author{Chi Ma}
\authornotemark[1]
\affiliation{%
  \institution{Meituan}
  \city{Beijing}
  \country{China}
}
\email{machi04@meituan.com}

\author{Xiao Yan}
\authornote{Corresponding authors.}
\affiliation{%
  \institution{Wuhan University}
  \department{Institute for Math and AI}
  \city{Wuhan}
  \country{China}
}
\email{yanxiaosunny@whu.edu.cn}

\author{Mincong Huang}
\author{Xiaoguang Li}
\affiliation{%
  \institution{Meituan}
  \city{Beijing}
  \country{China}
}
\email{huangmincong@meituan.com}
\email{lixiaoguang12@meituan.com}

\author{Ruidong Han}
\author{Bin Yin}
\affiliation{%
  \institution{Meituan}
  \city{Beijing}
  \country{China}
}
\email{hanruidong@meituan.com}
\email{yinbin05@meituan.com}

\author{Shangyu Chen}
\author{Xiang Li}
\affiliation{%
  \institution{Meituan}
  \city{Beijing}
  \country{China}
}
\email{chenshangyu03@meituan.com}
\email{lixiang245@meituan.com}

\author{Fei Jiang}
\author{Lei Yu}
\affiliation{%
  \institution{Meituan}
  \city{Beijing}
  \country{China}
}
\email{jiangfei05@meituan.com}
\email{yulei37@meituan.com}

\author{Chuan Liu}
\author{Wei Lin}
\affiliation{%
  \institution{Meituan}
  \city{Beijing}
  \country{China}
}
\email{liuchuan11@meituan.com}
\email{linwei31@meituan.com}

\author{Haowei Han}
\author{Xiaokai Zhou}
\affiliation{%
  \institution{Wuhan University}
  \department{School of Computer Science}
  \city{Wuhan}
  \country{China}
}
\email{haowei.han@whu.edu.cn}
\email{xiaokaizhou@whu.edu.cn}

\author{Bo Du}
\affiliation{%
  \institution{Wuhan University, School of Computer Science, National Engineering Research Center for Multimedia Software, Hubei Key Laboratory of Multimedia and Network Communication Engineering}
  \city{Wuhan}
  \country{China}
}
\email{dubo@whu.edu.cn}

\author{Jiawei Jiang}
\authornotemark[2]
\affiliation{%
  \institution{Wuhan University}
  \department{School of Computer Science}
  \city{Wuhan}
  \country{China}
}
\email{jiawei.jiang@whu.edu.cn}

\renewcommand{\shortauthors}{Yuxiang Wang et al.}

\begin{abstract}
Recommendation is crucial for both user experience and company revenue in Meituan as a leading lifestyle company, and generative recommendation models (GRMs) are shown to produce quality recommendations recently. However, existing systems are limited by insufficient functionality support and inefficient implementations for training GRMs in industrial scenarios. As such, we introduce \textbf{\method} as an efficient and scalable system for GRM training. Specifically, to handle real-time insertions/deletions of sparse embeddings, \method employs dynamic hash tables to replace static ones. To improve training efficiency, \method conducts dynamic sequence balancing to address the computation load imbalances among GPUs and adopts feature ID deduplication alongside automatic table merging to accelerate embedding lookup. Extensive experiments show that \method improves training throughput by $1.6 \times$ -- $2.4\times$ while
achieving good scalability when running over 100 GPUs. \method has been deployed for many applications in Meituan and is now handling hundreds of millions of requests on a daily basis. On the delivery platform, we observe a 1.22\% growth in user order volume and a 1.31\% enhancement in online PV\_CTR. 
\end{abstract}




\maketitle

\begin{figure}[t]
    \centering
    \includegraphics[scale=0.7]{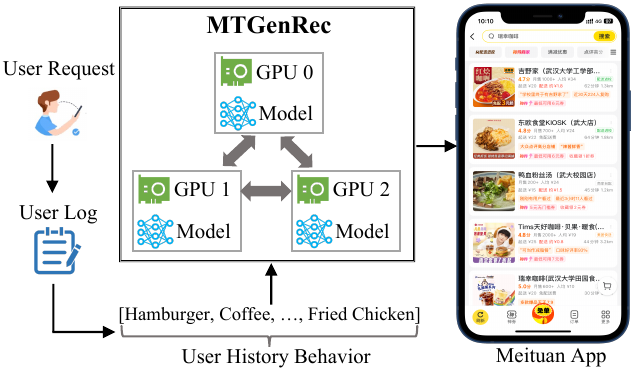}
    \caption{Recommendation workflow in Meituan.}
    \label{fig:application workflow}
\end{figure}

\section{Introduction}
As a leading e-commerce platform for lifestyle services, Meituan spans multiple business verticals such as delivery, travel, and health. In 2024, Meituan served over 770 million transacting users and processed more than 98 million orders at the daily peak. As shown in Figure~\ref{fig:application workflow},  recommendation systems~\cite{wang2017deep,zhou2018deep,naumov2019deep,jiang2019xdl,zhang2022picasso} are at the core of our business by effectively filtering history logs and personalizing services for users. High-quality recommendations are crucial to sustain high customer retention rates and enhance company revenue. Currently, our recommendation model utilizes terabytes of user action data and item metadata for training on a daily basis and produces billions of predictions for various services.

Inspired by generative language models~\cite{achiam2023gpt,radford2018improving,touvron2023llama,vaswani2017attention}, generative recommendation models (GRMs)~\cite{zhai2024actions,deng2025onerec,wang2023generative,yan2025unlocking,ji2024genrec,huang2025sessionrec,xu2025efficient} have gained significant attention from both academia and industry. Specifically, a GRM consists of two modules, i.e., \textit{sparse embedding tables} that map discrete categorical features to continuous embedding vectors, and a \textit{Transformer-based dense model} that predicts the next user action by modeling the action sequences of users. In contrast, earlier deep recommendation models (DRMs)~\cite{wang2017deep,zhou2018deep,naumov2019deep,xu2018deep} utilize multilayer perceptrons to capture the interactions between users and items. As shown in Figure~\ref{fig:beta fig}, GRM achieves superior accuracy compared with DRM, primarily due to its ability to model the entire action sequence of each user with self-attention. However, GRM has greater computation complexity due to its attention mechanism, which scales quadratically with sequence length, while DRM scales only linearly. 
Thus, to enjoy the good accuracy of GRMs and keep up with the high speed of training data generation, Meituan requires an efficient and scalable system for training GRMs.

TorchRec~\cite{ivchenko2022torchrec} is a state-of-the-art PyTorch-based system for training GRMs. It enables efficient computation of attention scores~\cite{dao2022flashattention} and enhances scalability through hybrid parallel strategies, including data and model parallelisms~\cite{shoeybi2019megatron,rasley2020deepspeed,huang2019gpipe,narayanan2019pipedream}.
To leverage TorchRec's benefits, we further develop our \method based on it.
However, we identify two fundamental limitations in TorchRec that hinder its application in industrial recommendation systems.
In particular, 
\ding{182} when TorchRec handles the sparse embedding tables of GRMs, its built-in static embedding tables cannot process dynamic embedding entries and table capacity expansion. Specifically, additional embeddings cannot be allocated to new users and items (e.g., merchants updating menus) in real-time from fixed-size static tables. Although default embeddings can be used in this case, model accuracy will be degraded. 
Moreover, static tables typically require pre-allocating more capacity than necessary to prevent overflow, leading to memory underutilization. \ding{183} For the embedding lookup process, TorchRec transmits same embedding entries multiple times due to duplicate embedding IDs in the user sequences, leading to a long communication time. Moreover, TorchRec requires labor-intensive manual configuration to merge embedding tables, which is a common and effective operation for accelerating lookups. The two limitations severely constrain the practical deployment of GRMs and compromise training efficiency.

\begin{figure}[t]
    \centering
    \includegraphics[scale=0.35]{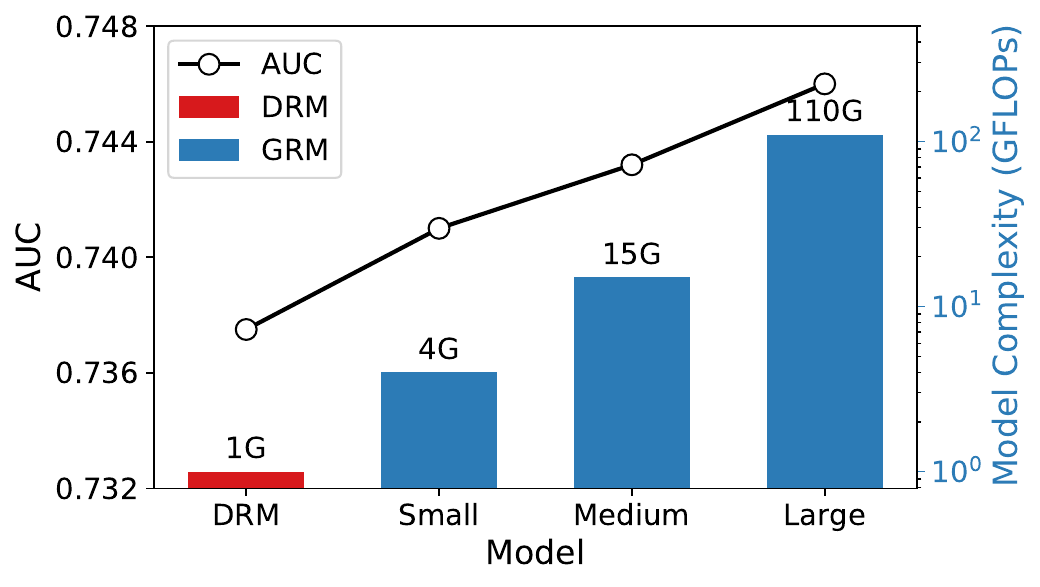}
    \caption{AUC (for accuracy, higher the better) and model complexity for DRM and GRM. Notably, an accuracy improvement of even 0.1\% is crucial for practical recommendation.}
    \label{fig:beta fig}
\end{figure}

To address these challenges, we propose a \textit{hash-based embedding table} for changeable IDs and \textit{two-stage ID deduplication} for lookup acceleration. The hash-based embedding table eliminates fixed-capacity constraints by mapping arbitrary feature IDs to embedding indices via a hash function.
We decouple the storage of keys (i.e., embedding index) and values (i.e., embedding vector).
This design prioritizes the expansion of the smaller key list over the larger value list, thereby enhancing the efficiency of table expansion.
Additionally, we adopt a two-stage ID deduplication operation to eliminate repetitive IDs before and after ID communication, thus avoiding the repetitive transmission of embedding vectors among devices. We also design a table configuration interface for automatic table merging, which combines multiple embedding tables into one according to defined patterns (e.g., the same embedding dimension).
This technique reduces the number of embedding lookup operators, and thus accelerating the process.
Besides the sparse embeddings, we observe that variable-length sequences cause significant load imbalances among GPUs. Since retaining complete sequences is crucial for GRM's accuracy, we cannot truncate or pad sequences as is done in LLMs. Instead, we propose a lightweight \textit{dynamic sequence batching} technique that dynamically adjusts batch sizes to ensure a more balanced workload across GPUs.

\begin{figure}[t]
    \centering
    \includegraphics[scale=0.4]{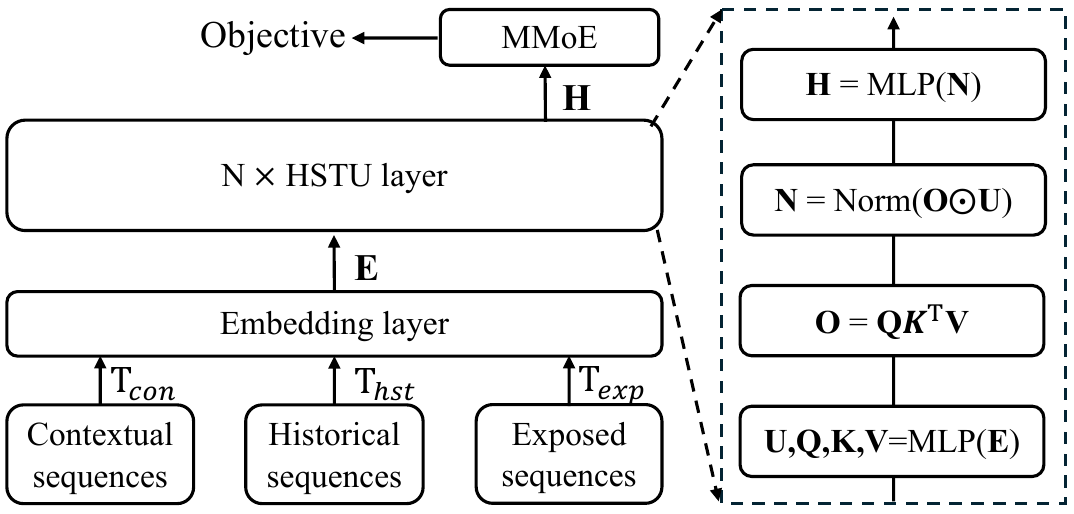}
    \caption{Model architecture of Meituan's GRM.}
    \label{fig:MTGR model}
\end{figure}

We evaluate \method and its designs using 200 million real user sequences generated over a week at Meituan. Scalability studies show that \method achieves close to linear speedup in training throughput when using more GPUs. Compared with TorchRec, \method reduces training time by $2.44 \times$. Ablation studies show that our designs are effective. In particular, dynamic sequence balancing achieves near-optimal load balance and improves system throughput by 1.75$\times$ compared to the baseline. ID deduplication significantly reduce network communication, yielding an improvement of 53\% in throughput.
Compared to the original DRM, GRM maintains similar training costs with \method, while achieving a 2.88pp increase in offline CTCVR GAUC, a 1.22\% growth in take-away orders, and a 1.31\% enhancement in PV\_CTR.

To summarize, we make the following contributions.

\squishlist
    \item We present \method as an efficient and scalable system for training generative recommendation models (GRMs) at Meituan.
    \item To address the challenges caused by GRM's sparse and dense models, we introduce dynamic embedding tables and ID deduplication for changeable feature IDs and lookup acceleration, and dynamic sequence batching for computation load balancing.
    \item We conduct extensive experiments to evaluate \method. The results show that \method is efficient and scalable, and our designs are effective in improving system efficiency.
\squishend

\section{\method System Overview}

\stitle{Model Architecture.} Figure~\ref{fig:MTGR model} illustrates the model architecture of the GRM in Meituan, which consists of four core components: input sequence, embedding layer, hierarchical sequential transduction units~\cite{zhai2024actions} (HSTU) layer, and multi-gate mixture-of-experts~\cite{ma2018modeling} (MMoE) layer. In particular, the input sequence $\mathbf{T}=[\mathrm{T}_{con},\mathrm{T}_{hst},\mathrm{T}_{exp}]$ includes contextual sequences $\mathrm{T}_{con}$ (i.e., user features), historical sequences $\mathrm{T}_{hst}$ (e.g., click and purchase actions), and exposition sequences $\mathrm{T}_{exp}$ (i.e., candidate items). Since the model cannot directly process discrete category data, the embedding layer $\phi_{emb}$ is used to convert the feature IDs into continuous embedding vectors $\mathbf{E}=\phi_{emb}(\mathbf{T}) \in \mathbb{R}^{F\times d}$, where $F$ and $d$ denote the number of features and embedding dimension, respectively.

Subsequently, we employ multiple HSTU layers -- a self-attention based Transformer~\cite{vaswani2017attention} architecture -- to learn user sequence embeddings from item interactions. As shown in Figure~\ref{fig:MTGR model}, each HSTU is equipped with four attention-based sub-layers for user feature extraction $\mathbf{U}$, item feature extraction $\mathbf{O}$, and feature interaction $\mathbf{H}$. The computation process is expressed as follows:
\begin{gather}
    \mathbf{U, Q, K, V} = \mathrm{Split}(\phi_{1}(\mathrm{MLP}(\mathbf{E}))), \\
    \mathbf{O}=\phi_{2}(\mathbf{Q}\mathbf{k}^T)\mathbf{V}, \\
    \mathbf{H}=\mathrm{MLP}(\mathrm{Norm}(\mathbf{O}\odot \mathbf{U})),
\end{gather}
where $\phi_1$ and $\phi_2$ are the SiLU~\cite{elfwing2018sigmoid} activation function. 
MMoE directs embeddings \(\mathbf{H}\) from the HSTU layer to multiple expert models. Each expert model is equipped with a gating network that learns the model's weights. We finally aggregate the output embeddings of the top-$k$ expert models. That is:
\begin{equation}
    y=\sum_{i=1}^{k}g_i(\mathbf{H})\cdot \mathrm{Expert}_i(\mathbf{H}).
\end{equation}
We select the cross entropy loss as the training objective.

\stitle{Distributed Training.} \method facilitates the distributed training of GRMs across multiple devices. 
We follow TorchRec's design and employ a hybrid parallel strategy: model parallelism for sparse models and data parallelism for dense models.
This design is driven by two considerations. First, data parallelism is impractical for the large embedding table, as no single device can store the entire embedding table. Second, dense models are relatively small, making data parallelism more suitable.
\method is trained end-to-end with synchronous execution to ensure model quality and accuracy. 
To further enhance training efficiency, we introduce a pipeline that preprocesses the next data batch concurrently with the model computation.
We detail the pipeline design in Appendix~\ref{adx:pipeline}.
Next, We provide an introduction to the workflow of \method as follow.

\begin{figure}[t]
    \centering
    \includegraphics[scale=0.4]{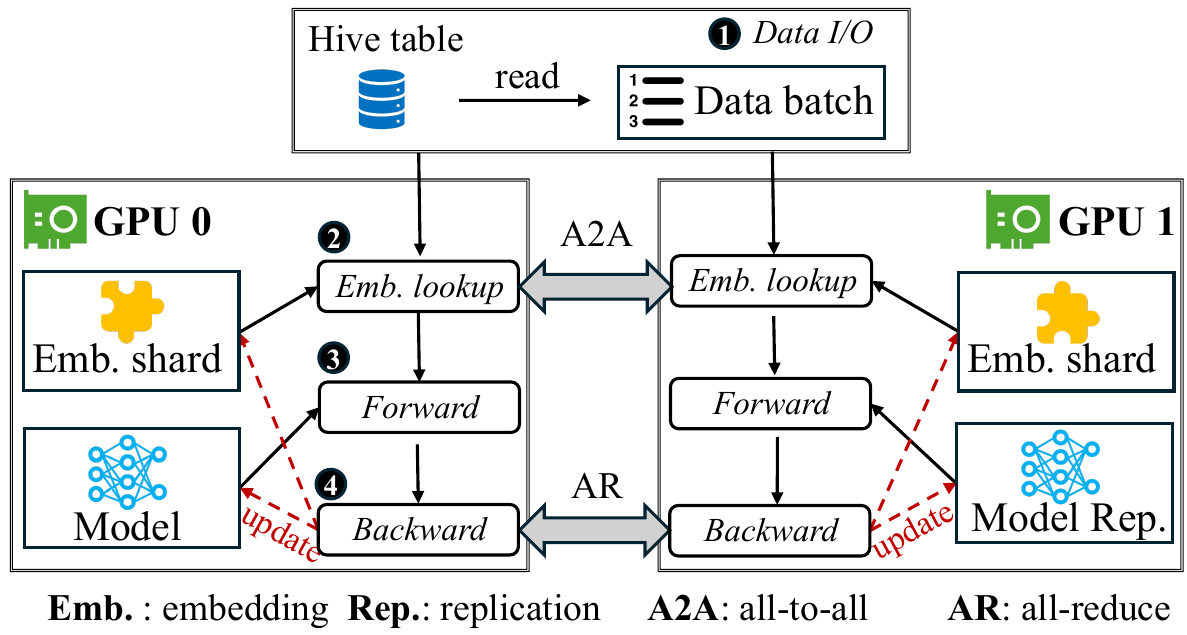}
    \caption{The workflow of \method for GRM training.}
    \label{fig:system workflow}
\end{figure}

\squishlist
\item \stitle{Data I/O.} As shown in Figure~\ref{fig:system workflow}, the training data is stored in partitioned Hive tables on HDFS, which utilizes a columnar storage format to optimize access speed and storage efficiency. Unlike DRM that pairs a user with a single item in a data sample, GRM uses a sequence-wise approach that pairs a user with multiple items. This design avoids the repeated computation of a user feature.
We provide an explanation in Appendix~\ref{apd:different from GRM and DRM}.

\item \stitle{Embedding Lookup.}
The embedding lookup uses All-to-all communication for cross-device embedding exchange. In particular, this process involves two data exchanges: First, devices transmit the required feature IDs and receive corresponding IDs from peer devices. Subsequently, a collection operation facilitates the exchange of retrieved embeddings between devices.

\item \stitle{Forward Computation.} After embedding lookup, we encode the input sequence through the dense model (i.e., HSTU and MMoE layers) and then calculate the loss function. Model parameters across all devices are consistently initialized by setting the same hyperparameter and random seed.

\item \stitle{Backward Update.} \method employs a heterogeneous parameter update strategy in the backward propagation. Specifically, sparse embedding tables aggregate gradients through model parallelism, where each device updates its locally stored embedding shards. Dense parameters employ All-Reduce communication for gradient synchronization before parameter update.
\squishend

\section{Core Designs for \method}

\subsection{Dynamic Embedding Table}
In a real production environment, the recommendation system must handle real-time insertions and deletions of sparse embedding entries (e.g., merchants adding or removing products). 
However, TorchRec cannot handle dynamic entries due to its static embedding tables (table capacity is fixed), and it typically uses default embeddings for feature IDs that exceed the table capacity. However, this method may lead to model accuracy degradation. Additionally, static tables require over-provisioning of initial capacity to anticipate potential expansion requirements, which results in inefficient memory utilization. These limitations motivate \method's implementation of hash-based embedding tables.

\stitle{Storage Layout.} 
\method employs a decoupled hash table architecture to separate keys and values into distinct structures. As shown in Figure~\ref{fig:hash table overview} (a), the lightweight \textit{key structure} maintains a mapping table of keys and pointers to embeddings, while the \textit{embedding structure} stores embeddings and auxiliary data (e.g., timestamps). This design allows for efficient table expansion by prioritizing the expansion of the smaller key structure over the larger embedding structure.
Additionally, we implement chunk-based allocation for the embedding structure. Specifically, multiple embeddings are stored within a single chunk, and insertions/deletions are managed on a per-chunk basis.
This approach enhances efficiency, as it requires only a single chunk allocation or deletion operation, rather than multiple individual operations.

\stitle{Hash Function and Collision Handling.} \method selects MurmurHash3~\cite{murmurhash} as the hash function to determine key indices due to its efficient hash value calculation. MurmurHash3 processes input IDs in 4-byte blocks through mixing operations, which maximize entropy and ensure avalanche effects from single-bit changes. The multi-stage mixing generates uniformly distributed hash values through iterative nonlinear transformations.
Hash collisions occur when different keys are mapped to the same index by a hash function.  Common techniques for handling collisions include chaining and open addressing. We employ open addressing due to its memory efficiency advantages over chaining. However, existing probing techniques are vulnerable to clustering under high load factors~\cite{bender2022linear}.

\begin{figure}
    \centering
    \includegraphics[scale=0.29]{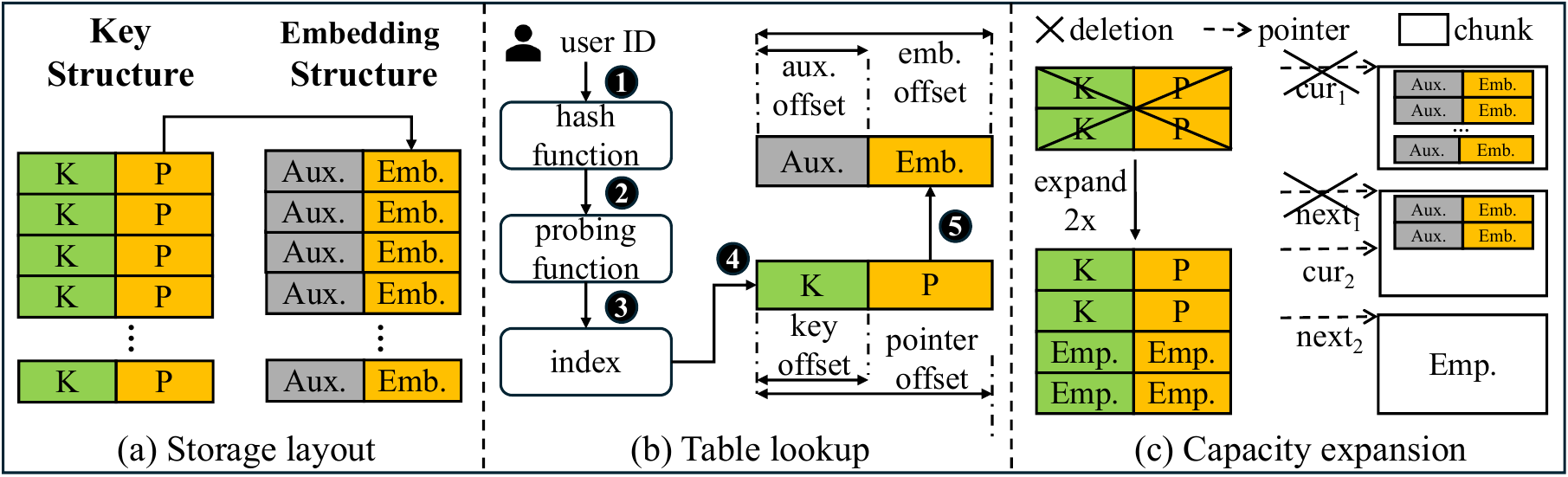}
    \caption{Dynamic embedding table in \method.}
    \label{fig:hash table overview}
\end{figure}

To address this issue, we propose a novel grouped parallel probing technique. This approach first ensures that the base step size $S$ is an odd number by performing a bitwise OR operation with 1, guaranteeing that all positions in the table can be accessed during probing. Next, we associate the step size with the key to prevent clustering caused by overlapping probing sequences of different keys. Finally, we multiply the step size by the thread group number, enabling distinct thread groups to probe different regions of the table. Formally, the computation is expressed as follows:

\begin{equation}
    S  = (k\%(M/\mathsf{threads}-1)+1 \mid  1)*\mathsf{threads},
\end{equation}
where $k$ refers to the key's value, $M$ represents the hash table size, $\mathsf{threads}$ indicates the number of thread groups, and $\mid$ denotes the bitwise OR operation. 
We use coprime conditions to prove that our grouped parallel probing technique can traverse all slots in the hash table. The proof is provided in Appendix~\ref{app: prrof}.

\begin{theorem}
    For a hash table of size $M=2^n$  and an odd probing step $S$, the probe sequence $\{h_{t}\}_{t=0}^{M\!-\!1}, h_{t}\!=\!(h_0\!+\!tS)\%M$ covers all $M$ distinct slots if and only if $S$ is odd. Formally:
        \begin{equation}
            \forall i,j\in\left [ 0, M-1 \right ] ,i\neq j \Rightarrow h_i\neq h_{j}.
        \end{equation}
\end{theorem}

\stitle{Table Lookup.} With the hash function and probing method in place, we can now implement dynamic table lookup. As shown in Figure~\ref{fig:hash table overview} (b), a hash value is computed for a given user ID, and probing identifies the index corresponding to this ID within the hash table (steps 1-3). Using the starting address of the key structure in memory, along with the index and pointer offset, we determine the embedding pointer $p$ for the user ID (step 4):
\begin{equation}
    p \!=\! st\_add \!+\! index * row\_offset\!+\!pointer\_offset,
\end{equation}
where $st\_add$ denotes the start address of the key structure.
Finally, we use this pointer, along with the embedding offset in the embedding structure, to extract the embedding vector (step 5).

\stitle{Capacity Expansion.} 
When the slots in the hash table approach full capacity, expansion becomes necessary.
\method triggers this expansion when the slot occupancy reaches critical levels (i.e., load factor > 0.75). As shown in Figure~\ref{fig:hash table overview} (c), the capacity increases in a power-of-two progression, doubling iteratively to meet the required capacity. The migration process transfers the original key structure to the expanded version, and the legacy memory is deallocated. Notably, our design prioritizes expanding the key structure while maintaining the chunk (i.e., embedding structure) capacity, thus optimizing expansion efficiency by avoiding bulk transfers of embedding data.
To manage the expansion of the embedding structure, \method maintains two chunks: a current chunk and a next chunk. When the current chunk's remaining capacity is insufficient for new embeddings, they are stored in the next chunk. Concurrently, we retire the filled chunk and allocate a new chunk as the next chunk to preserve the dual-chunk configuration. This pre-allocation mechanism ensures immediate storage availability for new embeddings during the expansion process.

\subsection{Embedding Table Merging}
Industrial recommendation systems usually merge multiple embedding tables into a unified table, which improves efficiency by fusing multiple lookup operations into one~\cite{jiang2019xdl,zhang2022picasso}.
However, the merging process requires developers to manually configure each embedding table, which becomes labor-intensive for massive tables. To address this issue, we propose an automated table merging technique.

\stitle{Automated Merging Table.}
We design a unified feature configuration interface, \texttt{FeatureConfig}, which enables automated table merging by defining parameters for each feature (e.g., \textit{feature name, embedding dimensions, and lookup tables}). \method then generates merging strategies, such as combining tables with identical embedding dimensions. To support dynamic embedding table merging, we introduce a \texttt{HashTableCollection} that uses feature configurations and performs pooling computations as needed. Developers only need to specify the required features, and \method automatically registers them and executes the table merging without requiring manual intervention. 

\stitle{Previous Solution.} Table merging can accelerate the lookup process, however, it introduces numerical overlap in the ID spaces of different tables. For example, user IDs and item IDs might share the same value in the merged table. TorchRec~\cite{ivchenko2022torchrec} addresses this issue using a row-wise offset mechanism for static table merging. In particular, each embedding table is assigned a unique row offset, which is calculated by summing the number of rows from all preceding tables. This offset is then added to the original ID to generate a globally unique ID, which is used to retrieve the corresponding embedding from the merged embedding table. 

\begin{figure}[t]
    \centering
    \includegraphics[scale=0.45]{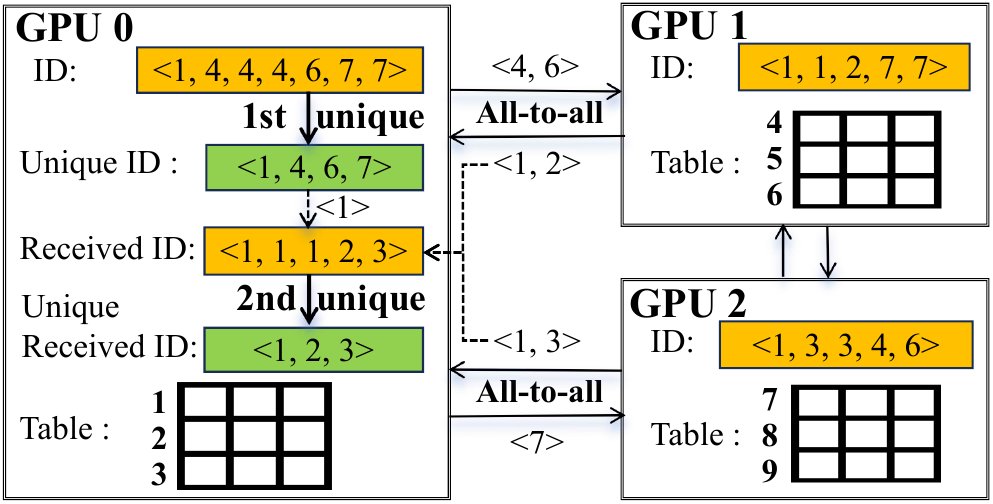}
    \caption{A running example of two-stage ID deduplication. We only visualize the ID deduplication process on GPU 0 and omit the embedding communication.}
    \label{fig: ID unique}
\end{figure}

\stitle{Our Solution.} Dynamic table merging also faces ID overlap issues. However, the row counts of dynamic tables are unpredictable in advance, which precludes the use of fixed offsets. To resolve this, we dynamically calculate offsets based on the number of tables to prevent ID overflow. Specifically, we use bitwise operations to combine the original ID with a feature table identifier to create a globally unique ID. We first compute the number of bits needed for encoding feature table indices as \( k = \lceil \log_2(m+1) \rceil \), using the high \( k \) bits of the 64-bit integer space as feature identifier bits. The highest bit is set to 0 to ensure the offset is a positive integer, while the remaining \( (64-1-k) \) bits are used to determine the maximum number of rows in the hash table.
For example, if we use 2 identifier bits to distinguish between 3 tables (since \(  \lceil \log_2(3+1) \rceil =2\)), the remaining 61 bits (64-1-2) are used to calculate the maximum row capacity. Consequently, the offset values for ``Table 2'' and ``Table 3'' are configured as $2^{59}$ and $2^{60}$, respectively, derived through
successive halving of the maximum row capacity.
The calculation formula of the ID $x$ in the \( i \)-th feature table is expressed as follows:
\begin{equation}
    \mathrm{ID}=(i \ll (63-k))\mid x ,
\end{equation}
where $\ll$ denotes shift left operation, $\mid$ is bitwise OR. A detailed running case is provided in Appendix~\ref{apd:table offset} for better understanding.

\subsection{Embedding Table Lookup Acceleration}
\label{sec:Lookup Acceleration}
Embedding tables are typically divided into shards distributed across devices. Each table lookup involves two all-to-all collective communications: one for ID exchange and another for embedding vector exchange. The communication cost is directly proportional to the number of features.
However, a sequence batch may contain numerous duplicate features IDs, and wasted bandwidth can render the lookup communication as performance bottlenecks.
To address this issue, we propose a two-stage ID deduplication technique aimed at reducing communication latency.

\stitle{Two-stage ID Deduplication.} As shown in Figure~\ref{fig: ID unique}, the first stage involves deduplicating feature IDs on each device before exchanging them, thereby reducing ID communication latency.
Despite this optimization, the ID exchange process can reintroduce duplicates due to different devices sending the same ID. Thus, the second deduplication stage is necessary to remove these newly generated duplicates.
ID deduplication can significantly alleviate embedding communication latency. If feature IDs are redundant, devices will transmit multiple duplicate embeddings, as they must return the corresponding embedding vectors based on the received feature IDs.
Notably, the primary objective of the two-stage ID deduplication is to reduce embedding communication latency, as ID communication has a relatively minor impact.

\begin{figure}[t]
    \centering
    \includegraphics[scale=0.35]{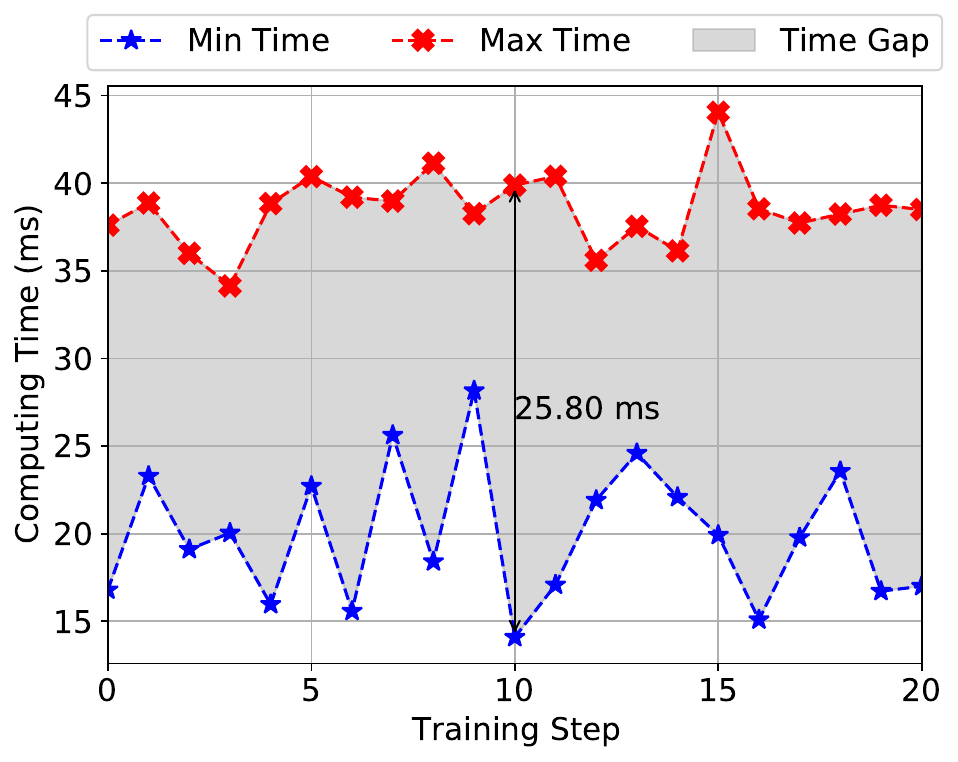}
    \caption{The visualization of imbalanced computational load. We train the \method on 8 GPUs and report the maximum and minimum GPU computation times across steps 0–20.}
    \label{fig:gpu computing time}
\end{figure}

\subsection{Sequence Balancing}
In addition to optimizing the sparse model of GRM, we observe that existing methods are still insufficient in optimizing the dense model.
User sequences exhibit a long-tail distribution: a small subset of highly active users generate exceptionally long sequences, while the most produce short ones. Previous DRMs are constrained by architectural limitations that enforce sequence truncation and padding for length alignment. In contrast, GRMs require preserving complete user sequences to achieve superior accuracy. This advantage arises because truncation risks removing semantically critical tokens, even when separated by many intermediate interactions. 

\stitle{Initial Design.} FlashAttention~\cite{dao2022flashattention} processes sequences by dividing them into chunks, where each chunk independently computes local attention. The final results are iteratively merged to cover the full sequence. This chunk-based design inherently supports variable-length sequences, motivating our initial approach of directly computing complete user sequences. 
However, this straightforward batching approach results in significant efficiency challenges. Figure~\ref{fig:gpu computing time} illustrates the GPU computation times per training step. The shaded region in the figure reveals prolonged synchronization delays (up to 25.8 ms) across devices, resulting in resource underutilization and reduced training efficiency. The primary cause is the unequal distribution of sequence lengths, which leads to a load imbalance issue among devices.

\begin{figure}[t]
    \centering
    \includegraphics[scale=0.55]{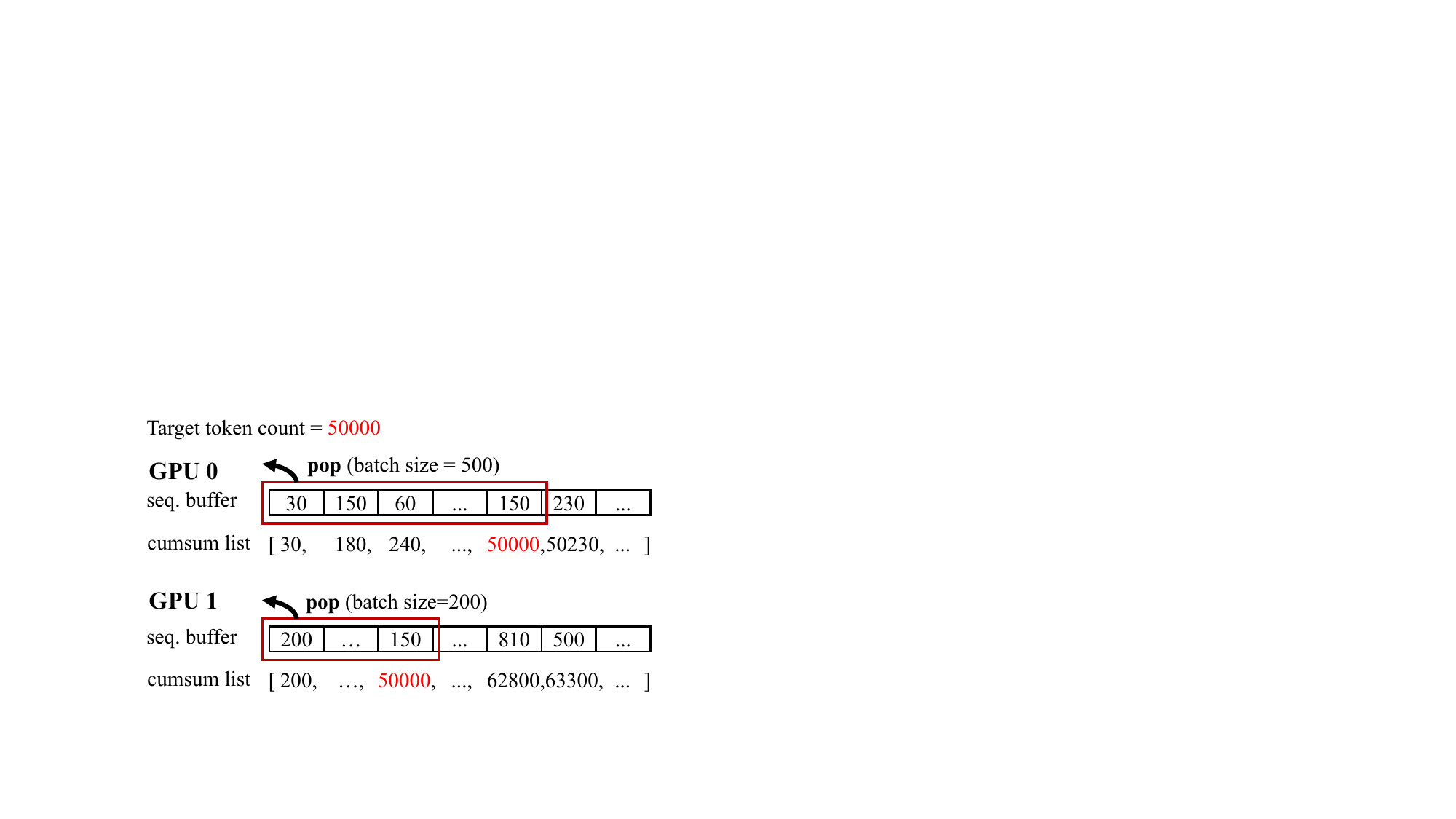}
    \caption{Running example of dynamic sequence batching. The number indicates the number of tokens in a sample.}
    \label{fig:global balance}
\end{figure}

\stitle{Dynamic Sequence Batching}.
To address the problem, as shown in Figure~\ref{fig:global balance}, we introduce a dynamic sequence batching technique. Specifically, each GPU maintains a buffer for storing sequence samples retrieved from Hive table chunks. 
We define the target token count as the product of average sequence length (i.e., 600) and the batch size.
We first compute token counts per sample in the buffer and calculate their cumulative sums. A binary search algorithm then identifies the optimal partition point where the cumulative sum most closely approaches the target token count. We finally obtain a balanced batch by outputting the first $k$ sequences in the buffer. If the total number of tokens in the buffer falls below the target, the remaining samples are merged into the subsequent buffer. This process iterates until all Hive table chunk samples are consumed.
A detailed algorithm is provided in Appendix~\ref{apd:batching algorithm}.

In data parallelism, gradients are typically averaged using All-Reduce operations. However, with dynamic sequence batching, directly averaging gradients may introduce biases, as each device computes gradients based on a varying number of samples. To address this issue, we implement All-to-all communication to synchronize batch sizes across devices, followed by performing a weighted average of gradients proportional to their respective batch sizes. This ensures numerical correctness in both loss computation and gradient update processes.

\stitle{Discussions.} Dynamic sequence batching provides two critical advantages: (1) improved device load balancing via token constraints, and (2) optimized hardware utilization via adaptive batch sizes. A fixed batch size strategy requires conservative configurations to prevent out-of-memory risks from clusters of extremely long sequences. In contrast, dynamic batching maximizes device memory utilization by loading batches near device memory limits. 
While more complex load balancing methods exist, such as sequence packing based on length distributions~\cite{krell2021efficient,kundu2024enhancing} or dynamic programming-based parallelism strategies~\cite{li2024demystifying}, we ultimately exclude these methods from our \method. This decision stems from the empirical observation that our lightweight dynamic sequence balancing achieves near-optimal workload balancing under current training conditions, rendering additional complexity unnecessary.

Besides the above optimizations for sparse and dense model, we incorporate a suite of additional techniques to implement a comprehensive GRM training system, including checkpoint resumption, mixed-precision training, gradient accumulation, and operator fusion. A detailed introduction is provided in Appendix~\ref{apd:implementations}. 

\section{Experimental Evaluations}

\subsection{Experiment Settings}
\stitle{Model Setups.} 
For GRM's dense model, we scale model architecture based on \textit{computational complexity}, resulting in GRM 4G and GRM 110G variants, where ``G'' corresponds to Giga Floating-Point Operations (GFLOPs) required per forward pass. For sparse models, we use the \textit{embedding dimension factor} to expand the embedding tables, including 1D, 8D, and 64D configurations.
The baseline 1D configuration chooses the widely adopted embedding dimensions~\cite{wang2017deep,xu2018deep,naumov2019deep}, typically ranging from 32 to 128 (with variations across different features).
Accordingly, 64D represents a 64$\times$ expansion of all embedding table dimensions relative to this 1D baseline.
The detailed model hyperparameters are provided in Appendix~\ref{apd:model complexity}.

\stitle{Baselines.} 
We compare \method with TorchRec~\cite{ivchenko2022torchrec}, the state-of-the-art GRM training framework proposed by Meta.
Actually, our initial design is developed based on TorchRec, and we further implement \method with the proposed optimizations.
We conduct experimental comparisons against TorchRec regarding both model accuracy and efficiency. 
We exclude alternative baselines as TorchRec is the only PyTorch-native training framework specifically optimized for PyTorch-based GRMs. 
Although other training systems such as TensorFlow are widely used in the industry, their framework currently cannot support GRM training.

\stitle{Dataset and Implementation.} We conduct experiments on 10 days of user logs from Meituan, with daily log data exceeding 200 million sequences. The average sequence length is 600, reaching a maximum length of 3,000. 
The detailed dataset statistics are provided in Appendix~\ref{apd:dataset}.
All models are trained from scratch using the same data.
To optimize sparse and dense features, we employ the Adam~\cite{kingma2014adam} optimizer.

\stitle{Environment.} Our experiments are conducted on a training node equipped with 8 GPUs by default. For scaling experiments, the cluster configuration extends to a maximum of 16 GPU nodes, each containing 8 GPUs. All GPUs are NVIDIA A100 SXM4 models with 80GB of memory. GPUs within a node are interconnected via NVLink with a bandwidth of 600 GB/s, while nodes are connected using Infini-Band interconnects with a bandwidth of 200 GB/s. 

\stitle{Evaluation Metrics.} We use \textit{Group AUC} (GAUC) to evaluate the accuracy of the model on CTR and CTCVR tasks and system \textit{throughput} to evaluate the training efficiency. 

\begin{figure}[t]
    \centering
    \includegraphics[scale=0.45]{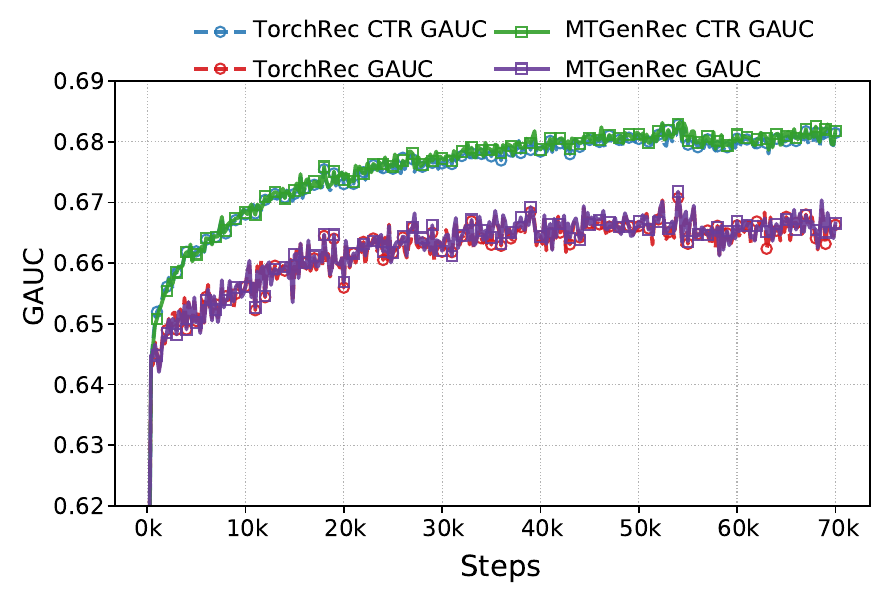}
    \caption{The GAUC results of CTR and CTCVR on GRM 4G 1D for TorchRec and \method, respectively.}
    \label{fig:exp auc}
\end{figure}

\begin{table}[t]
\caption{Comparison of offline and online performances on the Meituan Take-away application. We report the improved accuracy of GRM compared to DRM.}
\scalebox{0.85}{
    \centering
    \begin{tabular}{lcccc}
    \toprule
     & \multicolumn{2}{c}{Offline Metric diff} & \multicolumn{2}{c}{Online Metric diff} \\
    \cmidrule(lr){2-3} \cmidrule(lr){4-5}
     & CTR GAUC & CTCVR GAUC  & PV\_CTR & UV\_CTCVR \\
    \midrule
    GRM 4G & +0.0036 & +0.0154 & +1.04\% & +0.04\% \\
    GRM 110G & +0.0153 & +0.0288 & +1.90\% & +1.02\% \\
    \bottomrule
    \end{tabular}
    }
    \label{tab:online_metrics}
\end{table}

\subsection{Comparisons to Baselines}
\stitle{Accuracy and Business Benefits.} Figure~\ref{fig:exp auc} presents the CTR and CTCVR GAUC results obtained from training a GRM using \method and TorchRec. The results show that both \method and TorchRec ensures model correctness and training stability. To further validate the effectiveness of our GRM, we deploy \method on the Meituan Take-away application, conducting A/B testing with $2\%$ of the traffic. The comparison baseline is the most advanced online DRM (i.e., DLRM~\cite{naumov2019deep}), which has been continuously learning for 2 years. As shown in Table ~\ref{tab:online_metrics}, although the volume of training data is significantly lower, the offline and online metrics of GRM still greatly exceed those of the baseline.

\begin{table}[t]
\caption{Time (second) decomposition for GRM 4G 1D and GRM 110G 64D in TorchRec and \method, respectively.}
\scalebox{0.85}{
\begin{tabular}{llrrr}
\toprule
\textbf{Model}                & \textbf{Method} & \textbf{Lookup} & \textbf{Forward} & \textbf{Backward} \\
\midrule
\multirow{2}{*}{GRM 4G 1D}    & TorchRec        & 7.0 (s)             & 14.8 (s)             & 50.7 (s)              \\
                              & MTGenRec        & 3.2 (s)             & 14.0 (s)             & 32.6 (s)              \\
                              \midrule
\multirow{2}{*}{GRM 110G 64D} & TorchRec        & 95.2 (s)            & 27.8 (s)             & 177.7 (s)             \\
                              & MTGenRec        & 14.4 (s)            & 19.9 (s)             & 54.7 (s)     \\
\bottomrule
\end{tabular}
}
\label{tab:exp ablation time study}
\end{table}

\stitle{Time Decomposition.} Table~\ref{tab:exp ablation time study} presents execution times for GRM 4G 1D and 50G 64D over 100 training steps, detailing embedding lookup, forward, and backward phases. The results show that \method achieves faster execution than TorchRec across all phases. The efficiency improvements from sequence balancing increase substantially with model complexity. Larger embedding dimensions induce significant increases in the lookup and backward phases due to communication overhead and sparse parameter update latency.

\begin{figure}[t]
    \centering
    \includegraphics[width=\linewidth]{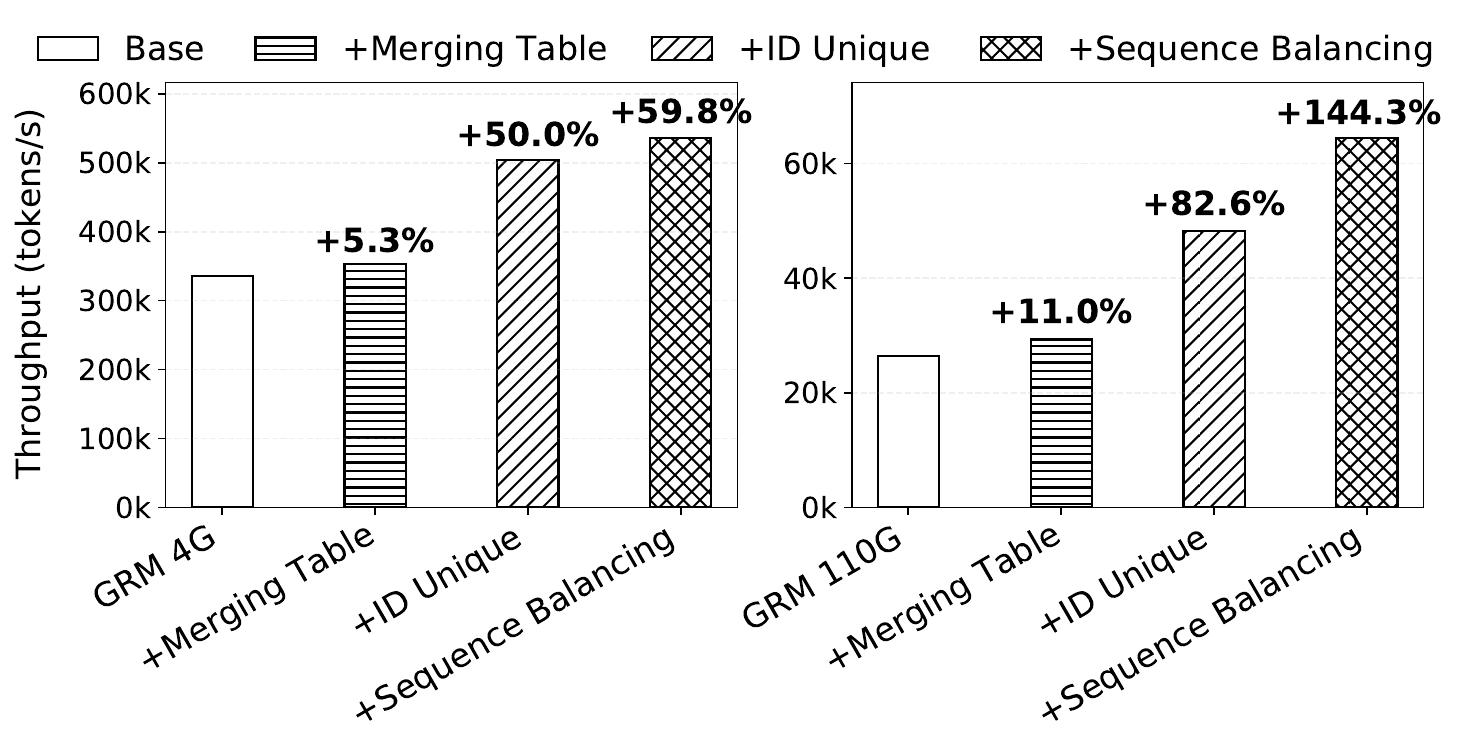}
    \caption{Ablation studies for GRM 4G 1D and GRM 110G 1D in \method. Compared with the baseline, \method achieves a 1.60 $\times$ to 2.44 $\times$ throughput improvement.}
    \label{fig:exp ablation study}
\end{figure}

\subsection{Evaluation of Individual Designs}

\stitle{Ablation Study.} Figure~\ref{fig:exp ablation study} presents an ablation study under 4G and 10G complexity settings by incrementally integrating different techniques in \method, including merging table, two-stage ID deduplication and sequence balancing.
The experimental results demonstrate all these techniques in \method can enhance system efficiency. Furthermore, we observe that performance gains are more pronounced with higher computational complexity. This trend reveals a positive correlation between \method’s efficiency improvements and computational complexity.

\begin{table}[t]
\centering
\caption{We report the comparison results in batch size and average GPU memory utilization before and after enabling sequence balancing. Numbers in parentheses denote improvements in GPU memory utilization. }
\scalebox{0.9}{
    \begin{tabular}{lrr}
\toprule
\textbf{Model} & \textbf{Batch size} & \textbf{GPU memory utilization} \\
\midrule
GRM 4G 1D       & 480 → 496             & 86.3\% → 95.7\% (+9.4)                  \\
GRM 110G 1D    & 80 → 116              & 75.3\% → 90.3\% (+15.0)    \\
\bottomrule
\end{tabular}
}
\label{tab: seq balance batch size}
\end{table}

\begin{figure}[t]
    \centering
    \includegraphics[scale=0.6]{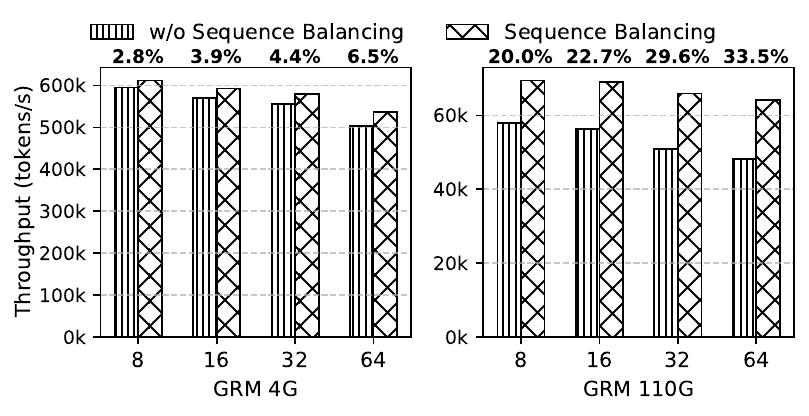}
    \caption{Throughput comparison for disabling and enabling sequence balancing on GRM 4G 1D and 110G 1D. We scale the system from 8 GPUs to 64 GPUs and report the gain values of sequence balancing at the top of figure.}
    \label{fig:exp throughput for balance}
\end{figure}

\stitle{Sequence Balancing Evaluation.} We evaluate sequence balancing on GRM 4G and 110G models, scaling from 8 to 64 GPUs. To ensure fairness, batch sizes are adjusted to maintain near-full GPU memory utilization under both enabled and disabled conditions. 
Figure~\ref{fig:exp throughput for balance} presents the throughput results, and table~\ref{tab: seq balance batch size} reports batch sizes and average GPU memory utilization, yielding four key observations. 
(1) Sequence balancing significantly improves throughput across all settings (e.g., an average improvements of 26.5\% for GRM 110G, peaking at 33.5\% on 64 GPUs). (2) Throughput gains scale with GPU count due to reduced synchronization delay from straggler devices. As the number of GPUs increases, the likelihood of encountering long sequences that induce computational delays also rises. (3) Throughput improvements correlate positively with computational complexity. Since sequence processing FLOPs scale quadratically with hidden embedding dimensions, higher computational complexity amplifies load imbalance, which makes sequence balancing increasingly critical. 
(4) Table~\ref{tab: seq balance batch size} shows that sequence balancing maximizes GPU memory utilization, whereas fixed-size batching requires conservative sizing to avoid OOM errors.

\begin{figure}[t]
    \centering
    \includegraphics[scale=0.4]{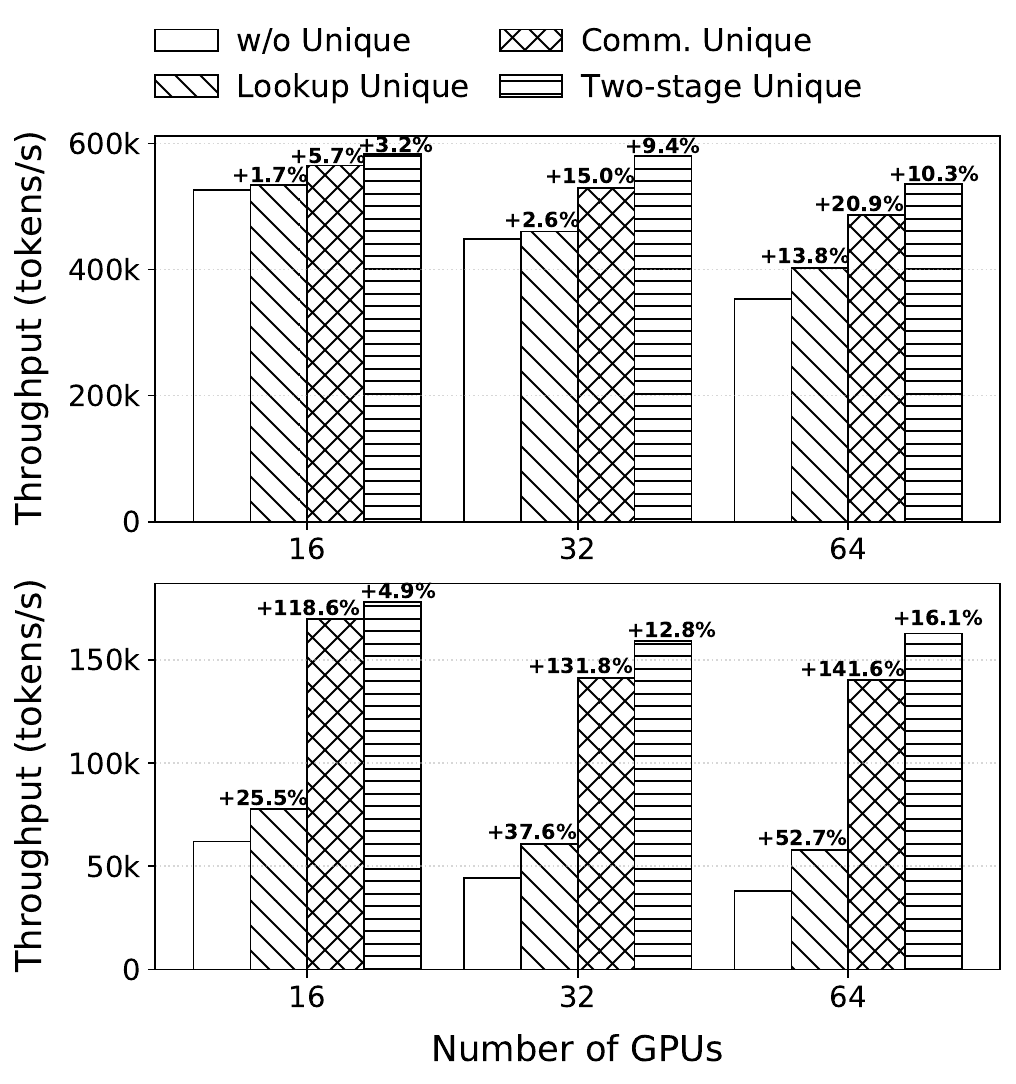}
    \caption{Throughput comparisons of two-stage ID deduplication for GRM 4G 1D (upper) and GRM 4G 64D (bottom). The improvements of the current strategy compared to the previous one are reported above the bars.}
    \label{fig:exp unique}
\end{figure}

\stitle{The Effect of Two-stage ID Deduplication.} To evaluate the impact of two-stage ID deduplication on system efficiency, we conduct experiments at 4 GFLOPs complexity with 1D and 64D embedding dimensions. We employ four strategies: (a) \textit{w/o unique} (no deduplication), (b) \textit{Comm. unique} (first-stage deduplication only), (c) \textit{Lookup unique} (second-stage deduplication only), and (d) \textit{Two-stage unique} (both stages). Starting with 16 GPUs, we incrementally increase the GPU count and report throughput results in Figure~\ref{fig:exp unique}. We get the following observations.
(1) The two-stage approach improves throughput by 1.1–3.7$\times$ over the baseline (a), highlighting its effectiveness in reducing ID communication and embedding communication. (2) ``Comm. unique'' outperforms ``Lookup unique'' because embedding communication is the primary source of latency, while hash table lookups are inherently fast. (3) The benefits of deduplication increase with embedding dimensions: higher-dimensional embeddings exacerbate network bandwidth waste, making ID deduplication more crucial for reducing communication latency.

\begin{table}[t]

\centering
\caption{Throughput comparison for MCH and \method. \textit{OOM} denotes out of memory and \textit{Gain} denotes the throughput improvement. ``-'' indicates that the comparison cannot be made due to OOM problems.}
\vspace{0.2cm}
\scalebox{0.85}{
\begin{tabular}{ccrrr}
\toprule
\textbf{Complexity}  & \textbf{Dim. Factor} & \textbf{MCH} & \textbf{\method} & \textbf{Gain} \\
\midrule
\multirow{3}{*}{4G}   & 1D                    & 392,731      & 579,649           & 47.59\%       \\
                     & 8D                    & 260,590      & 438,190           & 68.2\%       \\
                     & 64D                   & 80,857       & 155,832           & 92.7\%       \\
                     \midrule
\multirow{3}{*}{110G} & 1D                    & 44,719       & 68,993            & 54.28\%       \\
                     & 8D                   & 28,329       & 55,929            & 97.4\%       \\
                     & 64D                   & OOM          & 38,575            & -            \\
                     \bottomrule
\end{tabular}
}

\label{tab: exp mch vs. hash table}
\end{table}

\begin{figure}[t]
    \centering
    \includegraphics[width=\linewidth]{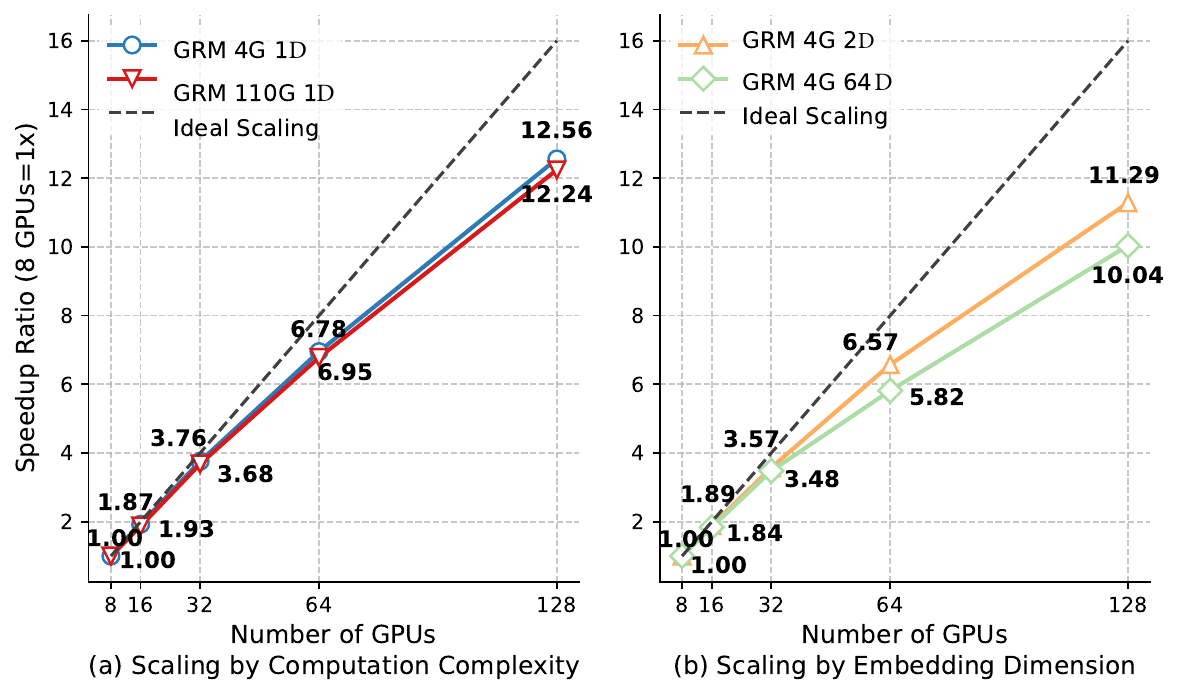}
    \caption{Comparison of scalability by computational complexity and embedding dimension. We use 8 GPUs configuration as the baseline for speedup ratio. The specific speedup ratio values are marked near the broken line.} 
    \label{fig:exp scaling}
\end{figure}

\stitle{The Effect of Dynamic Embedding Table.} We evaluate the proposed dynamic embedding table's effectiveness with TorchRec's Managed Collision Handling~\cite{ivchenko2022torchrec} (MCH) as the baseline. MCH remaps changeable IDs into a continuous space via a fixed-size mapping table. It employs binary search for ID localization and invokes eviction when a threshold is met. 
Experimental results are shown in Table~\ref{tab: exp mch vs. hash table}. Our findings reveal that the dynamic embedding table achieves 1.47–2.22$\times$ throughput improvement over MCH. This improvement stems from the grouped parallel probing technique, which accelerates hash collision resolution via multi-threaded parallel probing. Additionally, MCH encounters OOM issues at larger embedding dimensions due to pre-allocated memory for all tables. In contrast, the dynamic embedding table requires substantially less initial memory and allocates additional resources only when thresholds are exceeded.

\subsection{Scalability} 
We assess the system's scalability by examining deviation between actual speedup and ideal linear scaling for two factors: computational complexity and embedding dimensions. We fix either dimension or complexity while scaling the other, and use 8 GPUs as baseline up to 128 GPUs. The experimental results in Figure~\ref{fig:exp scaling} reveal three key findings.
(1) All systems exhibit sublinear scaling with increasing GPUs due to communication overhead. However, \method achieves 62.75\%–78.5\% of the ideal speedup at 128 GPUs, demonstrating its robust scalability. (2) Scaling computational complexity by 27.5$\times$ (4G 1D vs. 110G 1D) and embedding dimensions by 32$\times$ (4G 2D vs. 4G 64D) causes minor speedup degradation. The stability arises from sequence balancing, which mitigates load imbalance, and table merging and two-stage ID deduplication also reduce the lookup and communication latency. (3) Embedding dimensions impact speedup more significantly than computational complexity, as latency is dominated by sparse embedding communication, computation, and updates.

\section{Related Work}
\stitle{Systems for DRM Training.} Deep Recommendation Models (DRMs) combine sparse embeddings and dense neural networks to model each user-item individually. As DRM is widely used, many systems are developed to accelerate it~\cite{lian2022persia,jiang2019xdl,xie2020kraken,miao2021het,jiang2018linear,feng2020atbrg,fan2019mobius,zhao2020distributed,guo2021scalefreectr,smelyanskiy2019zion,krishna2020accelerating,wang2022merlin}. Specifically, Persia~\cite{lian2022persia} decouples  sparse embedding computation from dense model processing, training the embeddings asynchronously to improve throughput and the dense models synchronously  to ensure accuracy. PICASSO~\cite{zhang2022picasso} leverages the patterns of model architecture and data distribution to accelerate execution with packing, interleaving, and caching optimizations. XDL~\cite{jiang2019xdl} compresses the pair-wise communication among GPUs into a tree structure to accelerate embedding lookup and update. 
ScaleFreeCTR~\cite{guo2021scalefreectr} leverages host memory to store massive embedding tables and employs GPU-optimized synchronization to update the embeddings. Zion~\cite{smelyanskiy2019zion} and RecSpeed~\cite{krishna2020accelerating} mitigate the I/O bottlenecks for embedding communication by deploying additional NICs. 
However, these systems are typically designed based on specific assumptions about the access patterns of sparse embeddings (e.g., skewed towards frequently accessed embeddings), which limits their applicability compared to our dynamic embedding table. Additionally, these systems are implemented using TensorFlow, making them unsuitable for PyTorch-based GRM training.

\stitle{Systems for Transformer.} To train Transformer  efficiently, the focuses are attention computation and parallel strategies~\cite{ren2021zero,wang2023zero++,dao2022flashattention,lu2017flexflow,jia2018exploring,huang2019gpipe,shoeybi2019megatron,rasley2020deepspeed,li2024demystifying,wu2024bitpipe,rajbhandari2020zero}. FlashAttention~\cite{dao2022flashattention} reduces the accesses to slow GPU global memory with tiled computation and online recomputation fusion. 
FlexFlow~\cite{lu2017flexflow} and OptCNN~\cite{jia2018exploring} automatically search the optimal parallel strategy for model training by distributing model computation and parameters among the GPUs in different ways. GPipe~\cite{huang2019gpipe} employs pipeline parallelism by splitting each mini-batches into micro-batches to interleave the computation and communication across the micro-batches. Megatron-LM~\cite{shoeybi2019megatron} and DeepSpeed~\cite{rasley2020deepspeed} propose hybrid parallel strategies for large-scale model training by partitioning the model and data across the GPUs.
Hydraulis~\cite{li2024demystifying} adopts dynamic parallelism for each iteration according to the sequence length distribution to achieve workload balance. \method enjoys existing attention kernel optimizations for efficient computation but simple data parallel suffices for GRMs since their dense models are usually small. Moreover, \method optimizes the processing of sparse embeddings, which are not present in most Transformer-based models.     

\section{Conclusion}
We propose \method, a distributed training system designed to enhance scalability and efficiency of industrial generative recommendation models (GRMs). Our sparse model optimizations include dynamic embedding tables for changeable IDs, automated table merging, and two-stage ID deduplication. For dense models, dynamic sequence batching mitigates workload imbalances. The experiment results demonstrate substantially improved training efficiency without compromising accuracy, paving the way for broader adoption of GRMs in recommendation systems.

\begin{acks}
This work was sponsored by National Natural Science Foundation of China (62472327), Sichuan Clinical Research Center for Imaging Medicine (YXYX2402), Geological Hazard Prevention and Control Project of the Three Gorges Follow-up Work of the National Major Water
Conservancy Project Construction Fund (0001212024CC60003), and the Innovative Research Group Project of Hubei Province (2024AFA017). This work was supported by Meituan through the Efficient Training and Inference System for Generative Recommendation Models project (MT20250112053P).
\end{acks}

\bibliographystyle{ACM-Reference-Format}
\bibliography{reference}

\appendix

\section{The Differences between GRM and DRM} \label{apd:different from GRM and DRM}
Our model builds upon GR~\cite{zhai2024actions} by introducing a novel prediction head enhanced through MMoE. This design replaces GR’s MLP with MMoE, which employs task-specific gate networks to dynamically balance multiple objectives, improving accuracy in multi-task scenarios. Compared to MLP’s parameter-intensive tuning, MMoE reduces computational overhead by selectively activating only the top-k expert outputs. 
Additionally, our \method diverges from DRM in batch construction. As illustrated in Figure~\ref{fig:batch}, DLRM’s pairwise batches contain redundant user computations (e.g., "User 1" features are processed twice). In contrast, \method employs a sequence-wise approach that consolidates multiple user interactions into a single sample, which ensures each user feature is computed once. Notably, user features contain orders of magnitude more tokens than item features (e.g., 10,000 vs. 100), enabling \method to significantly accelerate training and improve scalability.

\begin{figure}[t]
    \centering
    \includegraphics[width=\linewidth]{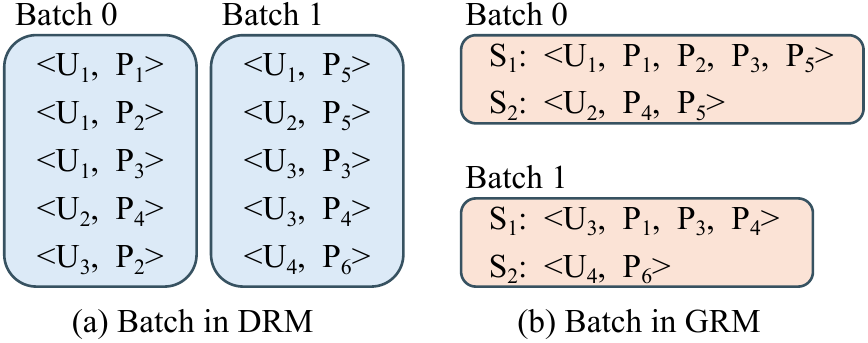}
    \vspace{-0.6cm}
    \caption{Batch structure of DRM and GRM. ``\textit{U}'' denotes users, ``\textit{P}'' denotes products. DRM and GRM use a pairwise and sequence-wise batch construction method, respectively.}
    \label{fig:batch}
\end{figure}

\begin{figure}[t]
    \centering
    \includegraphics[scale=0.5]{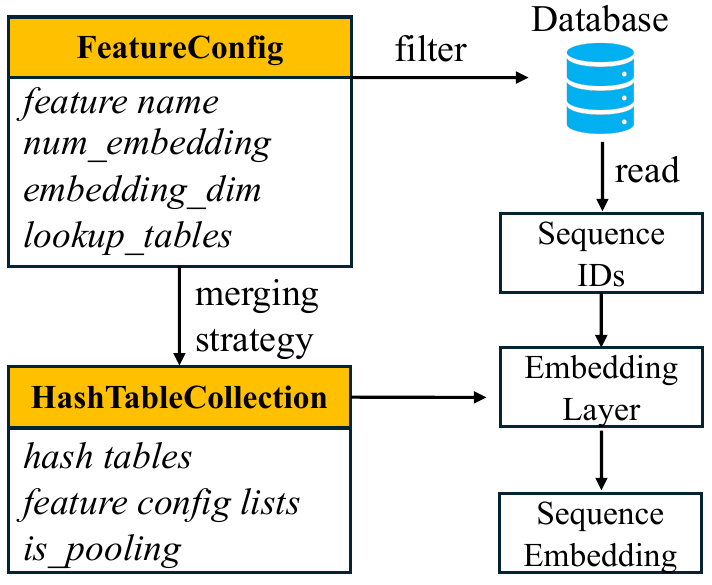}
    \vspace{-0.3cm}
    \caption{We achieve automatic table merging by two embedding configuration interfaces.}
    \vspace{-0.6cm}
    \label{fig:table merging interface}
\end{figure}

\section{Additionally Details for \method}
\subsection{Pipeline.}
\label{adx:pipeline}
We maximize parallelism through pipeline technology using three streams: \textit{copy}, \textit{dispatch}, and \textit{compute}. Specifically, the copy stream loads sequences from CPU to GPU, the dispatch stream performs table lookups based on IDs, and the compute stream handles forward computation and backward updates. 

\subsection{Proof}
\label{app: prrof}
We use coprime conditions and the characteristics of moduli to prove that our grouped parallel probing technique can traverse all slots in the hash table.
\begin{lemma}
    If $M=2^n$ and $S$ is odd, then $\mathrm{gcd}(S, M)=1 $. $\mathrm{gcd}$ denotes the greatest common divisor.
\end{lemma}
\begin{proof}
    $S$ contains no factor of 2 (since $S$ is odd), while $M$ has prime factorization $2^n$. Thus, $S$ and $M$ share no common prime factors, proving $\mathrm{gcd}(S, M)=1$.
\end{proof}
\begin{proof}
    Assume there exist $t_1,t_2\!\in \![0,M\!-\! 1]$, $t_1 \! \neq \! t_2$ such that:
    \begin{equation}
        h_{t_1} \equiv h_{t_2},
    \end{equation}
    This implies:
    \begin{equation}
        (t_1-t_2)\cdot S \equiv 0.
    \end{equation}
    By Lemma 1, $\mathrm{gcd}(S, M)=1$. According to Bézout's identity, $S$ has a multiplicative inverse modulo $M$. Multiplying both sides by this inverse yields:
    \begin{equation}
        t_1-t_2\equiv 0 .
    \end{equation}
    Since $t_1,t_2 \! \in \! [0, M\!-\! 1]$, the only solution is $t_1\!=\!t_2$, contradicting the initial assumption. Hence, all probe positions are unique.
    
\end{proof}

\begin{proof}
    By the pigeonhole principle, $M$ probes generate $M$ unique positions, exactly covering all slots:
    \begin{equation}
        \{h_{t}\}_{t=0}^{M-1}=\{ 0,1,2,...,M-1\}.
    \end{equation}
\end{proof}

\subsection{Interface for Automated Table Merging}
\label{apd:table merging interface}
Figure~\ref{fig:table merging interface} illustrates our unified feature configuration interface, \texttt{FeatureConfig}, which enables automatic table merging by defining a configuration entry for each feature based on its \textit{name}, \textit{embedding vocabulary size}, \textit{embedding dimension}, and \textit{lookup table}. Subsequently, \method reads data from the database according to feature names and generates a table merging strategy (e.g., merging tables with identical embedding dimensions). To support the merging of dynamic embedding tables, we design \texttt{HashtableCollection}, which generates a collection of hash tables based on the embedding configurations (e.g., \textit{feature name, embedding dimension}) within the \texttt{FeatureConfig} list and performs pooling operations if required. Finally, \method retrieves the corresponding embeddings for sequence IDs from the embedding layer. 

\subsection{Running Example for Dynamic Offset} \label{apd:table offset}
As shown in Figure~\ref{fig: table offset} (a), a unique row offset value is assigned to each embedding table (i.e., offset 100 for table 2), calculated by summing the number of rows from all preceding embedding tables. In contrast, as illustrated in Figure~\ref{fig: table offset} (b), we use 2 identifier bits to distinguish between three tables, while calculating the maximum row capacity of hash table as $2^{61}$ through utilization of the remaining 61 bits. 

\begin{figure}[t]
    \centering
    \includegraphics[scale=0.5]{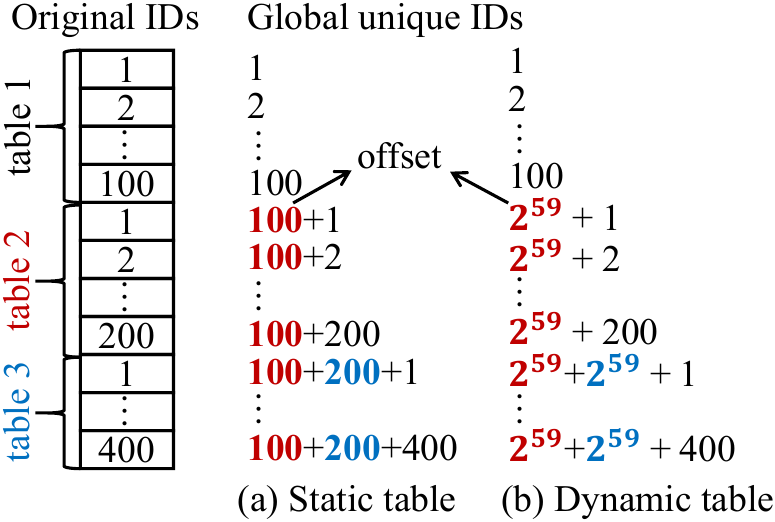}
    \vspace{-0.4cm}
    \caption{An example of offsets in static table and hash table.}
    \label{fig: table offset}
\end{figure}

\begin{algorithm}[t]
\caption{Dynamic Sequence Batching}
\SetAlgoLined
\DontPrintSemicolon
\KwIn{Target token count $N$, input chunks $\mathcal{C}$}

$\mathcal{Q} \gets \emptyset$, $C_i\in \mathcal{C}$, $i=0$\;
\While{$C_i \ne \emptyset$}{ 
    \If{$\texttt{sum}(\mathcal{Q})<N$}{
        $\mathcal{Q} \gets$ add all sequence in $C_i$\;
        $i++$\;
        }
    $S \gets$ cumulative sum of sequence token count in $\mathcal{Q}$\;
    $k \gets$ binary search in $S$ for closest value to $N$\;
    \If{$k<\texttt{len}(\mathcal{Q})$}{
        $\mathcal{B} \gets$ \texttt{pop}$(\mathcal{Q}[:k])$\;
        \textbf{yield} $\mathcal{B}$\;
    }
}
\KwOut{Balanced batch $\mathcal{B}$}
\label{alg: seq balance}
\vspace{-0.1cm}
\end{algorithm}

\subsection{Algorithm for Sequence Balancing} \label{apd:batching algorithm}
Algorithm~\ref{alg: seq balance} depicts \method's dynamic sequence batching. Each GPU maintains a sequence buffer $\mathcal{Q}$ populated from hive table chunks $C_i \in \mathcal{C}$. The target token count $N$ equals the batch size multiplied by the average sequence length. Token counts are aggregated into a cumulative sum $S$ over samples in $\mathcal{Q}$. A binary search finds the partition index $k$ where $S(k)$ best approximates $N$. The first $k$ sequences in $\mathcal{Q}$ form the output batch $\mathcal{B}$. Under-filled buffers accumulate samples until processing exhausts all hive chunk data.

\section{Implementation Details for \method} 
\label{apd:implementations}
\stitle{Checkpoint Resuming.} The core objective of checkpoint is to achieve persistent model preservation during training, enabling rapid recovery to previous training states. \method not only supports conventional checkpoint resuming but also facilitates preservation and loading in dynamic distributed environments (e.g., saving on 8 GPUs but loading on 16 GPUs). Existing systems address this challenge by saving model parameters from all devices in a single checkpoint, then redistributing parameters based on new device quantities during loading. However, this approach suffers from a critical flaw: each device must scan the entire checkpoint, resulting in redundant read operations. In contrast, \method implements a novel approach where each device independently preserves its own checkpoint. During loading, new devices locate required checkpoint files through modulo operations. For instance, when loading checkpoints saved from 8 GPUs onto 16 GPUs, both GPU 0 and GPU 8 load parameters from the checkpoint saved on the original GPU 0.  This design is grounded in the insight that distributed cluster scaling typically follows powers of two.

\stitle{Mixed Precision Training.} During training, we reduce computational resource consumption by converting dense models from FP32 to FP16 precision. For sparse embeddings, we dynamically partition hot feature sets based on access frequency. Specifically, for high-frequency accessed feature embeddings, we preserve embedding vectors in FP32 format to avoid quantization accumulation errors caused by frequent gradient updates. Conversely, low-frequency features employ FP16 storage and computation, significantly reducing memory footprint while accelerating table lookup operations.

\stitle{Gradient Accumulation.} \method enhances training stability by accumulating gradients across multiple training batches to mitigate large gradient variances caused by small batch sizes. For sparse embedding parameters, we first record activated embedding IDs and their corresponding gradient values within each batch. These gradients from identical IDs across multiple batches are accumulated and then updated collectively. Furthermore, we avoid full parameter updates for sparse embeddings, instead selectively updating only activated parts. This design simultaneously reduces memory waste and improves update efficiency. For smaller dense models, we also implement gradient accumulation followed by full parameter updates.

\stitle{Operator Fusion.} Inspired by Flash Attention, we implement customized operator fusion for the critical HSTU module in forward computation. Specifically, we partition the U, Q, K, and V matrices in Figure~\ref{fig:MTGR model} into tiles and process them sequentially in SRAM. Crucially, we use casual mask vectors to reduce unnecessary calculations by dynamically determining token skipping.

\begin{table}[t]
\caption{Dataset Statistics}
\vspace{-0.2cm}
\centering
\scalebox{0.85}{
\begin{tabular}{lccccc}
\toprule
Dataset & \#Users & \#Items & \#Exposure & \#Click & \#Purchases \\
\midrule
Train 
& 0.21 billion
& 4,302,391 
& 23.74 billion
& 1.08 billion
& 0.18 billion
\\
Test & 3,021,198 & 3,141,997 & 76,855,608 & 4,545,386 & 769,534 \\
\bottomrule
\end{tabular}
}
\label{tab:dataset}
\end{table}

\section{Experiment Details}
\subsection{Datasets} \label{apd:dataset}
We construct a training dataset based on logs from the real industrial-scale recommendation system in Meituan. The detailed dataset statistics are reported in Table~\ref{tab:dataset}. Unlike public datasets, our real dataset contains a richer cross features set and longer user behavior sequences.
In addition, the volume of our dataset is large, allowing complex models to achieve more adequate convergence during training.
For the offline experiments, we collect data over a $10$-day period. For the online experiments, in order to compare with the DLRM baseline that has been trained for over $2$ years, we constructed a longer-term dataset for the experiments, using data spanning more than $6$ months.

\begin{table}[t]
\setlength\tabcolsep{2pt}
\centering
\caption{Model hyperparameters in GRM. \textbf{Complexity} denotes computational complexity, and the unit is FLOPs.}
\scalebox{0.9}{
\begin{tabular}{lcccc}
\toprule
 & \textbf{Complexity} & \textbf{\# Emb. dim.} & \textbf{\# HSTU block} & \textbf{\# HSTU head} \\
\midrule
Small & 4G         & 512           & 3             & 2            \\
Large & 110G       & 1024          & 22            & 4           \\
\bottomrule
\end{tabular}
}
\label{tab:modle hyperparameters}
\end{table}

\subsection{Model Hyperparameters} \label{apd:model complexity}
We use two models: small and large GRM for all the experimental evaluations.
The detailed model hyperparameters are provided in Table~\ref{tab:modle hyperparameters}.
\clearpage

\end{document}